\documentclass[12pt]{article}
\usepackage{epsfig}
\usepackage{amsfonts}
\usepackage{mathtools} 
\usepackage{amscd}
\usepackage{latexsym}
\usepackage{amsmath,amssymb,amsthm}
\usepackage{verbatim}
\usepackage{setspace}
\usepackage{cite}
\usepackage{graphicx}
\usepackage[makeroom]{cancel}
\usepackage{mathtools}
\usepackage[all]{xy}
\usepackage{tikz}
\usepackage[textheight=9in, textwidth=6.5in, letterpaper]{geometry}
\usepackage{color}   %May be necessary if you want to color links
\usepackage{hyperref}
\usepackage{tikz-cd}
\usetikzlibrary{cd,decorations.pathmorphing}
\usepackage{authblk}
\hypersetup{
    colorlinks=true,  %set true if you want colored links
    linktoc=all,     %set to all if you want both sections and subsections linked
    linkcolor=black,  %choose some color if you want links to stand out
    citecolor=black,
    filecolor=black,
    urlcolor=black,
}

\numberwithin{equation}{section}

\newtheorem{Theorem}{Theorem}[section]
\newtheorem*{Theorem*}{Theorem}
\newtheorem{Corollary}[Theorem]{Corollary}
\newtheorem{Lemma}[Theorem]{Lemma}
\newtheorem{Proposition}[Theorem]{Proposition}
 { \theoremstyle{definition}
\newtheorem{Definition}[Theorem]{Definition}

\newtheorem{Example}[Theorem]{Example}
\newtheorem{Remark}[Theorem]{Remark} }
\def\p{\partial}

\def\<{\langle}
\def\>{\rangle}

\def\cO{\mathcal{O}}

\def\vp{\varphi}
\def\S{\mathcal{S}}
\def\L{\mathcal{L}}
\def\R{\mathcal{R}}
\def\ih{\frac{i}{\hbar}}
\def\SO{\text{SO}}
\def\so{\mathfrak{so}}

\def\M{\mathbf{M}}
\def\m{\mathcal{M}}

\def\be{\begin{equation}}
\def\ee{\end{equation}}
\def\beq{\be\begin{array}{c}}
\def\eeq{\end{array}\ee}
\def\bes{\be\begin{split}}
\def\ees{\end{split} \ee}
\def\bs{\begin{split}}
\def\es{\end{split} }
\def\nn{\nonumber}
\def\b{{\beta}}
\def\a{{\alpha}}

\def\e{{\epsilon}}

   \makeatletter
  \let\over=\@@over \let\overwithdelims=\@@overwithdelims
  \let\atop=\@@atop \let\atopwithdelims=\@@atopwithdelims
  \let\above=\@@above \let\abovewithdelims=\@@abovewithdelims
\renewcommand\section{\@startsection {section}{1}{\z@}%
                                   {-3.5ex \@plus -1ex \@minus -.2ex}%nn
                                   {2.3ex \@plus.2ex}%
                                   {\normalfont\large\bfseries}}

\renewcommand\subsection{\@startsection{subsection}{2}{\z@}%
                                     {-3.25ex\@plus -1ex \@minus -.2ex}%
                                     {1.5ex \@plus .2ex}%
                                     {\normalfont\bfseries}}

\title{Classical BV Pushforward}
\author[1,2]{Xin Geng}
\author[2]{Andrey Losev}
\author[1,2]{Vyacheslav Lysov}
\affil[1]{\it Center for Mathematics and Interdisciplinary Sciences\\ Fudan University, Shanghai, 200433, China}
\affil[2]{\it Shanghai Institute for Mathematics and Interdisciplinary Sciences (SIMIS)\\ Shanghai, 200433, China}

\begin{document}
\begin{titlepage}
\unitlength = 1mm

\vskip 1cm

\maketitle

\begin{center}

\begin{abstract}
We define the classical BV pushforward and establish its two main properties. First, we prove that it maps solutions of the classical master equation (CME) on the total space to solutions of the CME on the smaller space. Second, we prove that, for isolated critical points, it maps BV canonically equivalent solutions to BV canonically equivalent solutions. We illustrate the construction with examples and discuss its applications.

\end{abstract}
\vspace{1.0cm}
\end{center}
\end{titlepage}

\pagestyle{empty}
\pagestyle{plain}
\pagenumbering{arabic}
\tableofcontents

\section{Introduction}

The Batalin-Vilkovisky (BV)\cite{Batalin:1981jr, Batalin:1983ggl} pushforward is a powerful construction in quantum field theory and homological algebra. It formalizes the idea of integrating out fields and performing homotopy transfer. The key property of the BV pushforward is that it preserves the Quantum Master Equation (QME): it maps a class of QME solutions, modulo BV canonical transformations, to a class of solutions on a smaller space. 

For many practical applications, we are interested in the classical limit of the QME, the Classical Master Equation (CME), and its solutions. We can take the classical limit of the BV pushforward procedure to get a map between the CME solutions. Since the BV pushforward is an integral, the classical limit becomes the saddle point approximation. Hence, the induced CME solution is an extremization of the CME solution on the big space along the Lagrangian submanifold. We name such extremization as ``classical BV pushforward '' and adopt it as a definition.

Our definition of classical BV pushforward and its key property, as a map between CME solutions, requires just an odd symplectic structure. In contrast, existing proofs of the key property require extending CME solutions on the big space to full QME solutions, proving the existence of the integral, and establishing the classical limit for the integration result. Each of the listed steps could be obstructed, and we provide several obstruction examples in section \ref{sect_obstructions}.

In this paper, we define the classical BV pushforward and establish its two main properties. First, we prove that it maps solutions of the CME on the total space to solutions of the CME on the smaller space. Second, we prove that, for isolated critical points, it maps BV canonically equivalent solutions to BV canonically equivalent solutions. As part of the proof, we provide an explicit relation between the generators of BV canonical transformations on the total space and the smaller space. We illustrate the construction with examples and discuss its applications.

The paper is organized as follows. In Section 2 we review the BV pushforward and the relevant aspects of the QME and CME. In Section 3 we introduce the classical BV pushforward and prove the two main theorems. In Section 4 we present examples. In Section 5 we conclude and discuss open questions.

\section{BV pushforward}
In this section, we briefly review the BV pushforward following lecture notes \cite{mnev2019quantum}.

\subsection{Fiber BV integral}
\begin{Definition}\label{def_Cart_prod_SP}
    For a pair of SP manifolds $({\bf M},\Omega,\mu)$ and $(\mathcal{M}, \omega, \hat{\mu})$ the Cartesian product $M={\bf M}\times\mathcal{M}$ is an SP manifold $(M, \Omega + \omega, \mu\hat{\mu})$. Let $P: M \to {\bf M}$ be a projection on the first factor. We denote $\Delta$ and ${\bf \Delta}$ the BV Laplacian operators on total space $M$ and ${\bf M}$.   
\end{Definition}

\begin{Definition}
    For a Lagrangian submanifold $\mathcal{L} \subset \mathcal{M}$ we define {\it the fiber BV integral}
    \be\label{def_fiber_bv_int}
        P_\ast^{(\mathcal{L})} = \int_{\mathcal{L}}\sqrt{\hat{\mu}}\;\;: C^\infty( M) \to C^\infty({\bf M}).
    \ee
\end{Definition}
\begin{Remark}
    The integrand $\sqrt{\mu\hat{\mu}}\; f$ in (\ref{def_fiber_bv_int}) is a half-density on $M$. Khudaverdian \cite{khudaverdian2000semidensities} (see \cite{cattaneo2022canonical} for a recent review) showed that the fiber BV integral could be defined canonically for half-densities on an odd symplectic manifold. Namely, we do not need a compatible Berezinian.
\end{Remark}
\begin{Theorem}\label{thm_bv_stokes} {\bf (Stokes theorem for fiber BV integrals)}
    For a Cartesian product of SP manifolds in Definition \ref{def_Cart_prod_SP}, the following holds:
    \begin{itemize} \item The fiber BV integral (\ref{def_fiber_bv_int}) commutes with the BV Laplacian. Namely, 
    \be
        P_\ast^{(\mathcal{L})} \Delta = {\bf \Delta} P_\ast^{(\mathcal{L})} .
    \ee
    \item For a pair of Lagrangian submanifolds $\mathcal{L}_0$ and $\mathcal{L}_1$, homotopic in $\mathcal{M}$
    \be
        P_\ast^{(\mathcal{L}_0)} -P_\ast^{(\mathcal{L}_1)} = {\bf \Delta}(\ldots).
    \ee
    \end{itemize}
\end{Theorem}
One of the early proofs of the Theorem \ref{thm_bv_stokes} appeared \cite{Krotov:2006th}, while the modern ones can be found in \cite{mnev2019quantum, cattaneo2026bv}. Witten \cite{Witten:1990wb} used the odd Fourier transform to map the theorem \ref{thm_bv_stokes} into a similar theorem for the pushforward of differential forms (see \cite{qiu2011introduction} for a recent review).

\subsection{Quantum master equation}

We introduce the exponential function with parameter $\hbar \in \mathbb{R}^+$
\be\label{form_exp_param_bv_action}
    f = e^{\frac{i}{\hbar} \mathcal{S}}.
\ee
\begin{Remark}
    In the exponential parametrization, authors sometimes use $-1$ instead of $i$. However, for the later parts of our discussion, we need an $i$ factor so the saddle-point approximation for the integral of $f$ is a sum over all critical points rather than a single critical point with the minimal value of $\mathcal{S}$.
\end{Remark}
The BV-closeness condition $\Delta f = 0$ becomes  
\be
    \Delta e^{\frac{i}{\hbar}\mathcal{S}} =-\frac{1}{\hbar^2} e^{\frac{i}{\hbar}\mathcal{S}}\left( \frac12 \{\mathcal{S},\mathcal{S}\}-i\hbar \Delta \mathcal{S} \right) = 0.
\ee

\begin{Definition}
    For an odd symplectic space with compatible Berezinian $(\mathcal{M}, \omega, \mu)$, a formal series of even functions $\mathcal{S}\in C^\infty(\mathcal{M})[[-i\hbar]]$, known as the quantum (or effective) BV action, is a solution to the quantum master equation if
    \be
    \label{quantum_mast_eq}
    \frac12 \{\mathcal{S}, \mathcal{S}\} =i  \hbar \Delta \mathcal{S}. 
    \ee
\end{Definition}

\begin{Corollary}
    In terms of the expansion coefficients 
    \be\label{BV_action_hbar_expansion}
        \mathcal{S} =\mathcal{S}_0+(-i\hbar) \mathcal{S}_1+(-i\hbar)^2 \mathcal{S}_2+\ldots
    \ee
    the quantum master equation (\ref{quantum_mast_eq}) becomes a system of equations 
    \be\label{eq_qme_hbar_expansion}
    \begin{split}
        \{\mathcal{S}_0, \mathcal{S}_0\}& = 0, \;\; \{\mathcal{S}_0, \mathcal{S}_1\} = -\Delta \mathcal{S}_0, \\
        \{\mathcal{S}_0, \mathcal{S}_2\} &+ \frac12  \{\mathcal{S}_1, \mathcal{S}_1\}=-\Delta \mathcal{S}_1,\ldots
    \end{split}
    \ee 
\end{Corollary}

Given a QME solution $\mathcal{S}$ and an odd function $\mathcal{R}$ we can construct another QME solution $\mathcal{S}'$ by adding $\Delta$-exact term
\be
    e^{\ih \S'}=e^{\ih \S}+\Delta\left(\int_0^1 e^{\ih\S_t}\R_t dt\right).
\ee

\begin{Definition} \label{def_can_bv_transform} 
    Two solutions $\mathcal{S},\mathcal{S}'$ for the quantum master equation are equivalent $\mathcal{S}\sim \mathcal{S}'$, if there exists a canonical BV transformation: a family  $\mathcal{S}_t, \mathcal{R}_t \in C^\infty(\mathcal{M})[[-i\hbar]]$, parametrized by $t\in [0,1]$, such that $\mathcal{S}_0=\mathcal{S}$  and $\mathcal{S}_1=\mathcal{S}'$, and the following equation holds
    \be\label{def_BV_canonical_diff}
        \frac{d}{dt} \mathcal{S}_t=\{\mathcal{S}_t, \mathcal{R}_t\} -i\hbar\Delta \mathcal{R}_t.
    \ee
    The odd function $\mathcal{R}_t$ is called the generator of the canonical BV transformation.
\end{Definition}

\subsection{BV pushforward theorem}

\begin{Definition}
    For a QME solution $\mathcal{S}$ on $\bf{M}\times \mathcal{M}$ and a Lagrangian submanifold $\mathcal{L}\subset \mathcal{M}$, the induced BV action ${\bf S}$ on ${\bf M}$ is defined through the BV integral (\ref{def_fiber_bv_int}) of the exponential function. Namely,
    \be\label{eq_BV_induction}
       (-2\pi i \hbar)^{s(\mathcal{L})} e^{\frac{i}{\hbar} {\bf S} (\hbar)}= \int_{\mathcal{L}} \sqrt{\mu}\; e^{\frac{i}{\hbar} \mathcal{S}} .
    \ee
    We choose an integer $s(\mathcal{L})$ such that the induced action ${\bf S}(\hbar)$ is in $C^\infty({\bf M})[[-i\hbar]]$.
\end{Definition}

\begin{Remark}
    For perturbative evaluation of the BV integral (\ref{eq_BV_induction}), the induced action is in $C^\infty({\bf M})[[-i\hbar]]$ only if we have a single critical point. 
\end{Remark}

The Theorem \ref{thm_bv_stokes} for exponential functions implies that 
\begin{Theorem}\label{thm_BV_action_induction} 
    If $\mathcal{S}$ solves the QME on ${\bf M} \times \mathcal{M}$ then
    \begin{itemize}
        \item the induced BV action ${\bf S}$ solves the QME on $\bf{M}$.
        \item for a pair of homotopic Lagrangian submanifolds $\mathcal{L}$ and $\mathcal{L}'$ the corresponding induced actions ${\bf S}$ and ${\bf S}'$ are related by a BV canonical transformation.
    \end{itemize}  
\end{Theorem}

\subsection{Classical master equation}
\begin{Definition} 
    For an odd symplectic space $(\mathcal{M}, \omega)$, an even function $\mathcal{S} \in C^\infty(\mathcal{M})$, known as the classical BV action, is a solution to the classical master equation (CME) if 
    \be\label{cl_mast_eq}
        \{\mathcal{S}, \mathcal{S}\} = 0.
    \ee
\end{Definition}
\begin{Definition} \label{def_cl_can_bv_transform} 
    Two solutions $\mathcal{S},\mathcal{S}'$ for the CME are equivalent $\mathcal{S}\sim \mathcal{S}'$, if there exists a canonical BV transformation: a family  $\mathcal{S}_t, \mathcal{R}_t \in C^\infty(\mathcal{M})$, parametrized by $t\in [0,1]$, such that $\mathcal{S}_0=\mathcal{S}$  and $\mathcal{S}_1=\mathcal{S}'$, and the following equation holds
\be\label{def_cl_BV_canonical_diff}
    \frac{d}{dt} \mathcal{S}_t=\{\mathcal{S}_t, \mathcal{R}_t\}.
\ee
\end{Definition}

\begin{Lemma}\label{lemma_qme_to_cme}
    The $\mathcal{S}_0$-term of the QME solution $\mathcal{S} = \mathcal{S}_0+\hbar \mathcal{S}_1+\cO(\hbar)^2$ solves the CME. 
\end{Lemma}
\begin{proof}
    The classical master equation is the leading-order QME (\ref{eq_qme_hbar_expansion}).
\end{proof}
\begin{Remark}
    The construction of the higher-order terms $\mathcal{S}_{k}$ from $\mathcal{S}_0$ is called the BV-quantization. For further discussion of it, please see \cite{cattaneo2023bv}.
\end{Remark}

\subsection{Saddle point approximation for BV integral}
By the Theorem \ref{thm_BV_action_induction}, the induced action ${\bf S}(\hbar)$ solves the QME, while by the Lemma \ref{lemma_qme_to_cme}, its leading $\hbar \to 0$ order ${\bf S}_0$ solves the CME. The leading $\hbar\to 0$ order of the BV pushforward integral (\ref{eq_BV_induction}) is a sum over critical points of the exponents of the corresponding critical values of the BV-action. We formalize it below

\begin{Proposition}\label{prop_saddle_pt_bv_integral}
    The induced classical BV-action ${\bf S}_0$ is the extremization of the classical BV action $\mathcal{S}_0$ on $\bf{M}\times \mathcal{M}$ along the Lagrangian submanifold $\mathcal{L}\subset \mathcal{M}$.
\end{Proposition}
\begin{proof}
    We use the saddle point approximation to the BV integral
    \be\label{eq_sadd_pt_approx_bv_int}
        \int_{\mathcal{L}} \sqrt{\mu} e^{\frac i\hbar \mathcal{S}} = \int D\vp\; e^{\frac i\hbar \mathcal{S}_{\mathcal{L}}}=(2\pi i \hbar)^{s(\mathcal{L})} \sum_{\vp_{cl}} e^{\frac i\hbar \mathcal{S}_{\mathcal{L}}|_{\vp_{cl}} + \cO(\hbar^0)}. 
    \ee
    The sum is taken over the critical points (classical solutions) $\vp_{cl}$ for the classical BV-action restricted to the Lagrangian submanifold $\mathcal{S}_{\mathcal{L}} = \mathcal{S}_0|_{\mathcal{L}}$. We read off the induced classical BV-action from the exponential contribution to the saddle-point approximation. Namely,
    \be
        {\bf S}_0 = \mathcal{S}_{0}\Big|_{\mathcal{L},\vp_{cl}}.     
    \ee
\end{proof}

\subsection{Obstructions and peculiarities}\label{sect_obstructions}

Let us list the possible obstructions to using the BV pushforward to construct the classical induced action
\begin{enumerate}
    \item In the infinite-dimensional setting, the BV Laplacian $\Delta$ may be ill-defined, as in the example below. For more detailed discussion of this issue, see \cite{costello2015quantization}.

    \item Not every CME solution $\mathcal{S}_0$ could be extended to the QME solution. The leading quantum correction $\mathcal{S}_1$ is a cohomological problem
    \be
    D_{\mathcal{S}_0} \mathcal{S}_1 = -\Delta\mathcal{S}_0,\;\;\; D_{\mathcal{S}_0} = \{\mathcal{S}_0, \cdot\}. 
    \ee
 
    \item The BV pushforward is an integral, which may be ill-defined. Either we have a functional integral over an infinite-dimensional space, or the integral diverges.
    \item A perturbative evaluation (\ref{eq_sadd_pt_approx_bv_int}) for the BV pushforward integral may contain multiple saddle points. Hence, in the strict sense, the induced action is not a perturbative series in $i\hbar$.   
\end{enumerate}
\begin{Example}
    We consider a 1-dimensional mechanical system with periodic time $t\in S^1 = \mathbb{R}/\mathbb{Z}$ with a global translational symmetry $\delta \phi = \dot{\phi}$. The CME solution is 
    \be\label{eq_class_mech_translations}
        \mathcal{S} = \int_{S^1} dt\; \mathcal{L}(\phi, \dot{\phi}) + c\int_{S^1} dt\; \phi^\ast \dot{\phi}.  
    \ee
    We use a Fourier mode basis 
    \be
        \phi(t) = \sum_{n\in \mathbb{Z}}  e^{2\pi int} \phi_n,\;\; \phi^\ast(t) = \sum_{n\in \mathbb{Z}} e^{2\pi int} \phi^\ast_n,
    \ee 
    to rewrite the last term in (\ref{eq_class_mech_translations}) and define the BV Laplacian
    \be\label{eq_bv_lapl_mode_exp}
        \Delta = \sum_{n\in \mathbb{Z}} \frac{\p}{\p \phi_n} \frac{\p}{\p \phi^\ast_{-n}}.      
    \ee
    The BV Laplacian (\ref{eq_bv_lapl_mode_exp}) is ill-defined on the CME solution (\ref{eq_class_mech_translations}). Namely,
    \be
        \Delta\mathcal{S}  = c \sum_{n\in \mathbb{Z}} 2\pi in.  
    \ee
    
\end{Example}
\begin{Example}
    Let $\mathcal{M} = \Pi T^\ast\mathbb{C}^{1|1}$ with base coordinates $\phi, c$. We wish to quantize the CME solution 
    \be
        \mathcal{S}_0 = c \phi \phi^\ast.
    \ee
    The $\cO(\hbar)$ order of the QME gives us an equation for $\mathcal{S}_1$
    \be
        \{\mathcal{S}_0,\mathcal{S}_1\} =-\phi \phi^\ast \frac{\p \mathcal{S}_1}{\p c^\ast} -c\phi  \frac{\p \mathcal{S}_1}{\p \phi}+c \phi^\ast \frac{\p \mathcal{S}_1}{\p \phi^\ast}
        = \Delta \mathcal{S}_0 =-c\Rightarrow\S_1=\ln\phi.   
    \ee

    The solution $\S_1=\ln\phi$ is not polynomial, so the CME solution $\mathcal{S}_0$ cannot be extended to the QME solution as a polynomial.   
\end{Example}

\section{Classical BV pushforward}

\subsection{Induced action via BV extremization}

Motivated by Proposition \ref{prop_saddle_pt_bv_integral}, we introduce a definition. In this section and beyond, we discuss only the CME, so we drop the subscript $0$ we used to indicate the leading part of the QME solution. 
\begin{Definition}
    The induced classical BV-action ${\bf S}$ is the extremization of the classical BV-action $\mathcal{S}$ on $\bf{M}\times \mathcal{M}$ along the Lagrangian submanifold $\mathcal{L}\subset \mathcal{M}$. For multiple critical points or a critical submanifold, we have a set of induced classical BV-actions, one for each critical point, isolated or not. 
\end{Definition}

We choose global Darboux coordinates $\Phi^{a}, \Phi^\ast_{a}$ on $\bf{M}$. The tubular neighborhood theorem (Thm 4.4.6 in \cite{mnev2019quantum}) allows us to choose local Darboux coordinates $\vp^{\a},\vp^\ast_{\a}$ in a tubular neighborhood $\mathcal{U}\subset \mathcal{M}$ of $\mathcal{L}$, such that the Lagrangian submanifold $\mathcal{L} = \{\vp^\ast_\a =0\}$. The CME on the total space in these coordinates
\be\label{eq_BV_brack_Darboux_coord}
     \frac12 \{\mathcal{S},\mathcal{S}\}= (-1)^{|\Phi^a|}\frac{\p \mathcal{S}}{\p \Phi^a} \frac{\p \mathcal{S}}{\p \Phi_{a}^\ast}+(-1)^{|\varphi^\alpha|}\frac{\p\mathcal{S}}{\p\varphi^\alpha}\frac{\p\mathcal{S}}{\p\varphi^*_\alpha}=0.    
\ee
\begin{Remark}
    In (\ref{eq_BV_brack_Darboux_coord}) we used sign conventions from our earlier paper \cite{Losev:2023gsq}. 
\end{Remark}

In these coordinates an extremization of $\mathcal{S}(\Phi, \Phi^\ast, \vp, \vp^\ast)$ along $\mathcal{L}$ is just an extremization of $\mathcal{S}_{\mathcal{L}}(\Phi,\Phi^{\ast},\vp) = \mathcal{S}(\Phi, \Phi^\ast, \vp, 0)$ in $\vp$-variables. The critical point equations 
\be\label{eq_crit_point_equation}
    \frac{\p \mathcal{S}_{\L}}{\p \vp^\a}\Big|_{\vp_{cl}} =\frac{\p \mathcal{S}}{\p \vp^\a}\Big|_{\L,\vp_{cl}} = 0.
\ee

We denote by $\vp_{cl}=\vp_{cl}(\Phi, \Phi^\ast)$ solutions to the critical point equations (\ref{eq_crit_point_equation}). The classical BV pushforward is a function (collection of functions for multiple critical points) on ${\bf M}$ 
\be
    {\bf S}(\Phi, \Phi^\ast)=\mathcal{S}_{\L}(\Phi, \Phi^\ast, \vp_{cl})=\mathcal{S}(\Phi,\Phi^\ast, \vp_{cl}, 0).
\ee
\begin{Lemma}\label{3.2}
    For each critical point $\vp_{cl}$ the corresponding critical value ${\bf S}$ obeys
    \be
        \frac{\p {\bf S}}{\p \Phi^a} =\frac{\p \mathcal{S}}{\p \Phi^a}\Big|_{\L,\vp_{cl}},\;\;\;  \frac{\p {\bf S}}{\p \Phi^\ast_a} =\frac{\p \mathcal{S}}{\p \Phi^\ast_a}\Big|_{\L,\vp_{cl}},    
    \ee
\end{Lemma}
\begin{proof}
    By definition 
    \be
        {\bf S}=\mathcal{S}_{\L}\Big|_{\vp_{cl}} = \mathcal{S}\Big|_{\L,\vp_{cl}}.    
    \ee
    We use the chain rule, the critical point equation (\ref{eq_crit_point_equation}), and the explicit form of $\mathcal{L}$ to derive 
    \be
        \frac{\p {\bf S}}{\p \Phi^a} = \frac{\p \mathcal{S}_{\L}}{\p \Phi^a}\Big|_{\vp_{cl}}+\frac{\p \vp^\b_{cl}}{\p\Phi^a} \frac{\p \mathcal{S}_{\L}}{\p \vp^\b}\Big|_{\vp_{cl}}= \frac{\p \mathcal{S}_{\L}}{\p \Phi^a}\Big|_{\vp_{cl}}=\frac{\p \mathcal{S}}{\p \Phi^a}\Big|_{\L,\vp_{cl}}.    
    \ee
    The relation for $\Phi^\ast$ is derived in a similar manner.

\end{proof}

\subsection{First main theorem}

\begin{Theorem}
    For every critical point $\vp_{cl}$, the induced classical BV-action solves the CME on ${\bf M}$.
\end{Theorem}
\begin{proof}

    The CME equation on ${\bf M}$ for induced action simplifies using Lemma \ref{3.2}
    \be
    \begin{split}
        \frac12 \{{\bf S},{\bf S}\}_{\bf M} &= (-1)^{|\Phi^a|}\frac{\p {\bf S}}{\p \Phi^a} \frac{\p {\bf S}}{\p \Phi_{a}^\ast}        
        =(-1)^{|\Phi^a|} \frac{\p \mathcal{S}}{\p \Phi^a}\frac{\p \mathcal{S}}{\p \Phi_a^\ast} \Big|_{\L,\vp_{cl}} \\
        &=\frac12 \{\mathcal{S},\mathcal{S}\}\Big|_{\L,\vp_{cl}} - (-1)^{|\vp^\a|} \frac{\p \mathcal{S}}{\p \vp^\a}\frac{\p \mathcal{S}}{\p \vp_\a^\ast}  \Big|_{\L,\vp_{cl}} =0. 
    \end{split}
    \ee
    The first term on the second line vanishes due to CME on the total space. The second term vanishes due to the critical-point equation (\ref{eq_crit_point_equation}). 
    
\end{proof}
\begin{Remark}
    Our proof implies that if we relax the CME condition for $\mathcal{S}$, we will get a relation
    \be
        \{{\bf S},{\bf S}\}_{\M} =\{\mathcal{S},\mathcal{S}\}\Big|_{\L,\vp_{cl}}. 
    \ee
\end{Remark}

\subsection{Second main theorem}

\begin{Lemma}\label{lemma for thm2} For a non-degenerate critical point $\vp_{cl}$
    \be
        \{ {\bf S},\vp_{cl}^\a\}_{\M}-(-1)^{|\varphi^\a|}\p_{\vp^\ast_\a} \mathcal{S}\Big|_{\vp_{cl},\L} =0.
    \ee
\end{Lemma}
\begin{proof}
We use $\Big|$ as the abbreviation for $\Big|_{\varphi_{cl},\mathcal{L}}$. The derivative of the critical point equation
\be\label{eq_derivative_crit_pt_equation}
    \frac{\p\mathcal{S}}{\p\varphi^\a}\Big|=0\Rightarrow\frac{\p}{\p\Phi^a}\left(\frac{\p\mathcal{S}}{\p\varphi^\a}\Big|\right)=\frac{\p^2\mathcal{S}}{\p\Phi^a\p\varphi^\a}\Big|+\frac{\p\varphi_{cl}^\b}{\p\Phi^a} \frac{\p^2\mathcal{S}}{\p\varphi^\b\p\varphi^\a}\Big|=0.
\ee
We use Lemma \ref{3.2} to simplify
\be\label{eq_derivatives_simplified}
    \{{\bf S},\varphi_{cl}^\alpha\}_\M=(-1)^{|\Phi^a|}\left(\frac{\p{\bf S}}{\p\Phi^a}\frac{\p\varphi^\alpha_{cl}}{\p\Phi^*_a}+\frac{\p{\bf S}}{\p\Phi^*_a}\frac{\p\varphi^a_{cl}}{\p\Phi^a}\right)=(-1)^{|\Phi^a|}\left(\frac{\p\mathcal{S}}{\p\Phi^a}\Big|\frac{\p\varphi^\alpha_{cl}}{\p\Phi^*_a}+\frac{\p\mathcal{S}}{\p\Phi^*_a}\Big|\frac{\p\varphi^\alpha_{cl}}{\p\Phi^a}\right).
\ee
We multiply (\ref{eq_derivatives_simplified}) by a matrix of second derivatives at a critical point and use (\ref{eq_derivative_crit_pt_equation}) and its version for $\Phi^\ast$-derivative to remove the second derivative matrix
\be\nn
\begin{split}
    \{&{\bf S},\varphi_{cl}^\alpha\}_\M\cdot\frac{\p^2\mathcal{S}}{\p\varphi^\alpha\p\varphi^\beta}\Big|=(-1)^{|\Phi^a|+1}\left(\frac{\p\mathcal{S}}{\p\Phi^a}\Big|\frac{\p^2\mathcal{S}}{\p\Phi^*_a\p\varphi^\beta}\Big|+\frac{\p\mathcal{S}}{\p\Phi^*_a}\Big|\frac{\p^2\mathcal{S}}{\p\Phi^a\p\varphi^\beta}\Big|\right)\\
    &=(-1)^{|\vp^\beta|+1}\frac{\p}{\p\varphi^\beta}\left((-1)^{|\Phi^a|}\frac{\p\mathcal{S}}{\p\Phi^a}\frac{\p\mathcal{S}}{\p\Phi^*_a}\right)\Big|
    =(-1)^{|\vp^\beta|+1}\frac{\p}{\p\varphi^\beta}\left(\frac{1}{2}\{\mathcal{S},\mathcal{S}\}-(-1)^{|\varphi^\alpha|}\frac{\p\mathcal{S}}{\p\varphi^\alpha}\frac{\p\mathcal{S}}{\p\varphi^*_\alpha}\right)\Big|\\
    &=(-1)^{|\vp^\beta|+|\varphi^\alpha|}\frac{\p}{\p\varphi^\beta}\left(\frac{\p\mathcal{S}}{\p\varphi^\alpha}\frac{\p\mathcal{S}}{\p\varphi^*_\alpha}\right)\Big|=(-1)^{|\varphi^\alpha|+|\vp^\beta|}\left(\frac{\p^2\mathcal{S}}{\p\varphi^\beta\p\varphi^\alpha}\Big|\frac{\p\mathcal{S}}{\p\varphi_\alpha^*}\Big|\right)\\
    &=(-1)^{|\varphi^\alpha|}\left(\frac{\p\mathcal{S}}{\p\varphi_\alpha^*}\Big|\frac{\p^2\mathcal{S}}{\p\varphi^\alpha\p\varphi^\beta}\Big|\right).
\end{split}
\ee
We used the Leibniz formula in the second equality, the BV bracket on the total space in the third, the CME on the total space in the fourth, the critical point equation (\ref{eq_crit_point_equation}) in the fifth, and rearranged terms and derivatives in the last.

The last relation implies the Lemma's statement for non-degenerate critical points, when the matrix of second derivatives is of maximal rank
\be 
    rk\;\frac{\p^2\S}{\p\vp^\alpha\p\vp^\beta}\Big|_{\varphi_{cl},\mathcal{L}} = max.
\ee

\end{proof}

\begin{Example}
    The Lemma \ref{lemma for thm2} may fail for non-isolated critical points. We consider ${\bf M} = \Pi T^\ast \mathbb{R}$ and $\mathcal{M}=\Pi T^\ast \mathbb{R}^{1|1}$ with $\varphi^1$ even and $\varphi^2$ odd. The CME solution on the total space
    \be\nn
        \mathcal{S}=\frac{1}{2}(\varphi^1)^2+\varphi^1+\Phi\varphi^1\varphi^*_2.
    \ee
    The restriction of the CME solution to Lagrangian submanifold $\L=\{\varphi^*_1=0,\varphi^*_2=0\}$ is
    \be\nn
        \mathcal{S}_{\mathcal{L}}=\frac{1}{2}(\varphi^1)^2+\varphi^1.
    \ee
    The critical point equations are
    $$\begin{aligned}
        &\frac{\p\mathcal{S}_{\mathcal{L}}}{\p\varphi^1}=\varphi^1+1=0\Rightarrow \varphi^1_{cl}=-1,\\
        &\frac{\p\mathcal{S}_{\mathcal{L}}}{\p\varphi^2}=0\Rightarrow \varphi^2_{cl} \in \mathbb{R}^{0|1}.
    \end{aligned}$$
    Hence we have a critical submanifold $\mathbb{R}^{0|1}$ parametrized by values of $\varphi^2_{cl}$. For each critical point, the induced action is the same
    \be\nn
        {\bf S}=-\frac{1}{2}.
    \ee
    One of the statements in Lemma \ref{lemma for thm2} is violated. Namely 
    $$\begin{aligned}
        &\{{\bf S},\varphi^1_{cl}\}_{\M}-(-1)^{|\vp^1|}\p_{\varphi^*_1}\mathcal{S}|_{\varphi_{cl},\mathcal{L}}=\{-\frac{1}{2},-1\}_\M=0,\\
        &\{{\bf S},\varphi^2_{cl}\}_{\M}-(-1)^{|\vp^2|}\p_{\varphi^*_2}\mathcal{S}|_{\varphi_{cl},\mathcal{L}}=\{-\frac{1}{2},0\}_\M+\Phi\cdot\varphi^1|_{\varphi_{cl},\mathcal{L}}=-\Phi\neq 0.
    \end{aligned}$$
  
\end{Example}

\begin{Theorem}\label{thm_BV_induction_canonical_transform}
    The classical BV pushforward for an isolated critical point maps a BV canonically equivalent pair to a BV canonically equivalent pair.
\end{Theorem}
\begin{proof}
    We can perform a finite BV canonical transformation as a sequence of infinitesimal transformations. Hence, it is sufficient to prove the theorem for an infinitesimal BV canonical transformation. We perform an infinitesimal canonical transformation on the total space with generator $\mathcal{R}$ 
    \be
    \begin{split}
        \mathcal{S}'& = \mathcal{S}+\e \{\mathcal{S},\mathcal{R}\} \\
        &=\mathcal{S}+\e(-1)^{|\Phi^a|}\left(\p_{\Phi^a} \mathcal{S}\p_{\Phi^\ast_a}\mathcal{R}+\p_{\Phi^\ast_a} \mathcal{S}\p_{\Phi_a}\mathcal{R} \right) +\e(-1)^{|\varphi^\alpha|}\left(\p_{\vp^\alpha} \mathcal{S}\p_{\vp^\ast_\a}\mathcal{R}+\p_{\vp^\ast_\a} \mathcal{S}\p_{\vp_\a}\mathcal{R} \right).
    \end{split} 
    \ee
    The restriction of the action to $\L=\{\vp^\ast = 0\}$
    \be
        \mathcal{S}'_{\L} = \mathcal{S}_{\L}+\e \{\mathcal{S},\mathcal{R}\}\Big|_{\L}.
    \ee
   
    Given a critical point $\vp_{cl}$ in the undeformed case, we expect the critical point $\hat{\vp}_{cl}=\vp_{cl}+\e \delta \vp$. The induced action at critical point $\hat{\vp}_{cl}$
    \be\label{eq_two_induced_actions}
    \begin{split}
        {\bf S}'&=\mathcal{S}'_{\L}(\Phi, \Phi^\ast, \vp_{cl}+\e\delta\vp)=\mathcal{S}_{\L}(\Phi, \Phi^\ast, \vp_{cl}+\e\delta\vp)+\e \{\mathcal{S},\mathcal{R}\}\Big|_{\vp_{cl},\L}+\cO(\e^2)\\
        &=\mathcal{S}_{\L}(\Phi, \Phi^\ast, \vp_{cl})+\e \delta \vp^\a\; \p_{\vp^\a}\mathcal{S}_{\L}\Big|_{\vp_{cl}} +\e \{\mathcal{S},\mathcal{R}\}\Big|_{\vp_{cl},\L}+\cO(\e^2)\\
        &={\bf S} +\e \{\mathcal{S},\mathcal{R}\}\Big|_{\vp_{cl},\L}+\cO(\e^2).
    \end{split}    
    \ee
    We use $\Big|$ as the abbreviation for $\Big|_{\varphi_{cl},\mathcal{L}}$. We introduce 
    \be
        {\bf R} =\mathcal{R}(\Phi, \Phi^\ast, \vp_{cl})\Big|_{\mathcal{L}}.
    \ee
    It obeys the relation, similar to Lemma \ref{3.2} (and a similar relation for $\Phi^\ast$-derivative) 
    \be
    \begin{split}
        \p_{\Phi^a} {\bf R}&= \p_{\Phi^a}\mathcal{R}\Big| + \frac{\p \vp_{cl}^\b}{\p\Phi^a}  \p_{\vp^\b}\mathcal{R}\Big|.
    \end{split}    
    \ee
    An infinitesimal canonical transformation on ${\bf M}$
    \be
        \e \{{\bf S}, {\bf R}\}_{\M} = \e(-1)^{|\Phi^a|}\left(\p_{\Phi^a} {\bf S}\p_{\Phi^\ast_a}{\bf R}+\p_{\Phi^\ast_a} {\bf S}\p_{\Phi_a}{\bf R} \right).     
    \ee
     For our special choice of coordinates, we evaluate
    \be\nn
    \begin{split}
        &\e\{\mathcal{S},\mathcal{R}\}\Big|=\e(-1)^{|\Phi^a|}\left(\p_{\Phi^a} \mathcal{S}\p_{\Phi^\ast_a}\mathcal{R}+\p_{\Phi^\ast_a} \mathcal{S}\p_{\Phi^a}\mathcal{R} \right)\Big| +\e(-1)^{|\vp^\a|}\left(\p_{\vp^\a} \mathcal{S}\p_{\vp^\ast_\a}\mathcal{R}+ \p_{\vp^\ast_\a} \mathcal{S}\p_{\vp^\a}\mathcal{R} \right)\Big|\\
         &=\e(-1)^{|\Phi^a|}\p_{\Phi^a} {\bf S}\left(\p_{\Phi^\ast_a}{\bf R}-\frac{\p \vp_{cl}^\b}{\p\Phi_a^\ast}  \p_{\vp^\b}\mathcal{R}\Big|\right)\\
         &\qquad\qquad+\e(-1)^{|\Phi^a|}\p_{\Phi^\ast_a} {\bf S}\left(\p_{\Phi^a}{\bf R}-\frac{\p \vp_{cl}^\b}{\p\Phi^a}  \p_{\vp^\b}\mathcal{R}\Big|\right)+\e(-1)^{|\vp^\a|}\p_{\vp^\ast_\a} \mathcal{S}\p_{\vp^\a}\mathcal{R} \Big|\\
         &=\e\{{\bf S}, {\bf R}\}_{\M} -\e(-1)^{|\Phi^a|}\left(\p_{\Phi^a} {\bf S}\frac{\p \vp_{cl}^\b}{\p\Phi^\ast_a}  \p_{\vp^\b}\mathcal{R}\Big|+\p_{\Phi^\ast_a} {\bf S}\frac{\p \vp_{cl}^\b}{\p\Phi^a}  \p_{\vp^\b}\mathcal{R}\Big|\right)
         +\e(-1)^{|\vp^\a|}\p_{\vp^\ast_\a} \mathcal{S}\p_{\vp^\a}\mathcal{R} \Big|\\
         &=\e \{{\bf S}, {\bf R}\}_{\M} -\e\left(\{ {\bf S},\vp_{cl}^\a\}_{\M}-(-1)^{|\vp^\a|}\p_{\vp^\ast_\a} \mathcal{S} \Big|\right)\p_{\vp^\a}\mathcal{R} \Big|=\e \{{\bf S}, {\bf R}\}_{\M}.
    \end{split}    
    \ee
    In the last equality, we used Lemma \ref{lemma for thm2}. Hence, using equation (\ref{eq_two_induced_actions}) we proved that 
    \be\nn
        {\bf S}' -{\bf S}=\e\{\S,\R\}\Big| +\cO(\e^2)=\e\{{\bf S},{\bf R}\}_{\M}+\cO(\e^2).
    \ee    
    Namely, the classical BV pushforward of a BV canonically equivalent pair $(\S,\S')$ is a BV canonically equivalent pair $(\bf S,\bf S')$ on ${\bf M}$ with the generator $\bf R=\R\Big|$.
\end{proof}

\section{Examples}
\begin{Example}
    We consider $\SO(n+1)$ rotations acting on $\mathbb R^{n+1}$ and we decompose $\mathbb R^{n+1}=\mathbb R^n\times\mathbb R$ with the corresponding linear coordinates $(\Phi^a,\varphi)$.
    The Lie algebra $\so(n+1)$ is also decomposed as a linear space into $\so(n)\oplus\mathbb R^n$ with the corresponding linear coordinates $C^{ab}, B^a$, where $C^{ab}$ is anti-symmetric. This data gives us a CME solution on total space $M=\Pi T^*(\mathbb{R}^{n+1}\times\Pi\so(n+1))$ in Darboux coordinates
    $$\mathcal{S}=\Phi^a\Phi_a+\varphi^2+Q+F,$$
    $$Q=C^{ab}\Phi_a\Phi^*_b+B^a\varphi\Phi^*_a-B^a\Phi_a\varphi^*,$$
    $$F=\delta_{bc}C^{ab}C^{cd}C^*_{ad}+\delta_{ac}C^{bc}B^aB^*_b-B^aB^bC^*_{ab}.$$
    We raise and lower indices using the $\delta_{ab}$. We define $\mathcal{M} = \Pi T^\ast \mathbb{R}$ with Darboux coordinates $\vp, \vp^\ast$. The restriction of BV action to $\L=\{\varphi^*=0\}$ gives us
    $$\mathcal{S}_\L=\Phi^a\Phi_a+\varphi^2+C^{ab}\Phi_a\Phi^*_b+B^a\varphi\Phi^*_a+F.$$
    Extremization with respect to $\varphi$ gives us a single critical point
    $$\varphi_{cl}=-\frac{1}{2}B^a\Phi^*_a.$$
    The induced action 
    $${\bf S}=\Phi^a\Phi_a+C^{ab}\Phi_a\Phi^*_b+F+\frac{1}{4}B^aB^b\Phi^*_a\Phi^*_b $$
    is a CME solution on ${\bf M}=\Pi T^*(\mathbb R^n\times\Pi \so(n+1))$ describing $\SO(n+1)$ action on $\mathbb R^n$ which preserves the classical action $\Phi^a\Phi_a$. The quadratic term $\frac{1}{4}B^aB^b\Phi^*_a\Phi^*_b$ (bivector) indicates the on-shell closure of the symmetry algebra. The bivector structures naturally emerge when refining on-shell symmetries, as studied in \cite{Alexandrov:2007pd,Losev:2023gsq}. 
\end{Example}

\begin{Example}
    For the previous example, we consider the BV canonical transformation with generator $\mathcal{R}_t=b^a\vp^*\Phi^*_a$.
    It is equivalent to the following coordinate transformation:
    $$\Phi^a\to\Phi^a-b^a\vp^*,\quad\vp\to\vp+b^a\Phi^*_a.$$
    So the new action is
    $$\S'=\Phi^a\Phi^a+\vp^2+Q'+F-b^ab^b\Phi^*_a\Phi^*_b-B^ab^b\Phi^*_a\Phi^*_b+2\delta_{ac}C^{bc}b^a\Phi^*_b\vp^*,$$
    $$Q'=C^{ab}\Phi_a\Phi^*_b+(B^a+2b^a)\varphi\Phi^*_a-(B^a+2b^a)\Phi_a\varphi^*,$$
    $$F=\delta_{bc}C^{ab}C^{cd}C^*_{ad}+\delta_{ac}C^{bc}B^aB^*_b-B^aB^bC^*_{ab}.$$
    The restriction of BV action to $\L=\{\varphi^*=0\}$ gives us
    $$\mathcal{S}'_\L=\Phi^a\Phi_a+\varphi^2+C^{ab}\Phi_a\Phi^*_b+(B^a+2b^a)\varphi\Phi^*_a+F-b^ab^b\Phi^*_a\Phi^*_b-B^ab^b\Phi^*_a\Phi^*_b.$$
    Extremization with respect to $\varphi$ gives us a single critical point
    $$\varphi'_{cl}=-\frac{1}{2}(B^a+2b^a)\Phi^*_a.$$
    The induced action is
    $${\bf S}'=\Phi^a\Phi_a+C^{ab}\Phi_a\Phi^*_b+F+\frac{1}{4}B^aB^b\Phi^*_a\Phi^*_b.$$
    It is exactly the same with $\bf S$, since the generator of BV canonical transformation ${\bf R}=\R|_{\L,\vp_{cl}}=0$ on ${\bf M}$ is trivial.
\end{Example}

\section{Conclusions and open questions}

In some applications, we perform the functional integral over an infinite-dimensional space of functions to obtain an effective theory with a finite number of degrees of freedom. 
One such example was studied by authors in \cite{LOSEV2026105723}.

In many cases the space (of fields) $\m$ is infinite-dimensional, while the space (of parameters) ${\bf M}$ is finite-dimensional. In the task of quantization, we would like to evaluate the dependence of the quantum partition function on parameters. The standard way is to start with a CME solution on ${\bf M} \times \mathcal{M}$, perform BV quantization of a solution to the QME (very hard to do), and then do a quantum BV pushforward to ${\bf M}$ (also hard because of the infinite-dimensional integral). 

We propose an alternative approach that seems much easier and more tractable. We take a CME solution on $\M\times\m$ and perform classical BV pushforward from $\M\times\m$ to $\M$ to get a CME solution on the space of parameters $\M$. Note that we do not need to do a functional integral; we just need to solve differential equations. To get a QME solution on the space of parameters ${\bf M}$, we perform BV quantization of the CME solution. Since we do BV quantization on a finite-dimensional space, the procedure is much more tractable and easier.

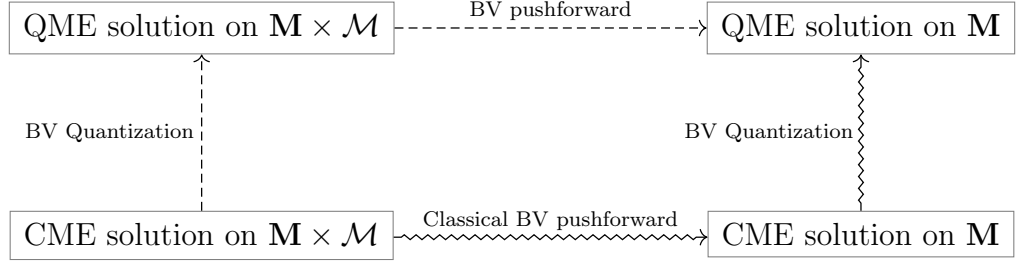
\begin{center}
        \begin{tikzcd}[column sep=10em,row sep=5em,cells={nodes={draw=gray}}]
            \text{QME solution on $\M\times\m$} \ar[r,"\text{BV pushforward}",dashrightarrow]  & \text{QME solution on $\M$}  \\
            \text{CME solution on $\M\times\m$} \ar[r,"\text{Classical BV pushforward}",squiggly] \ar[u,"\text{BV Quantization}",dashrightarrow] & \text{CME solution on $\M$} \ar[u,"\text{BV Quantization}",squiggly,]
        \end{tikzcd}
\end{center}

\section*{Acknowledgments}
The authors are grateful to Alberto Cattaneo and Maxim Zabzine for valuable discussions and to the organizers of the workshop “Higher Structures in Quantum Field Theory” for providing a space for the discussions.

\bibliography{BV_push_ref}{}
\bibliographystyle{utphys}

\end{document}